\documentclass[twoside,leqno,twocolumn]{article}

\usepackage[letterpaper]{geometry}
\usepackage{subcaption}
\usepackage{multirow}
\usepackage{makecell}

\usepackage{ltexpprt}
\usepackage{hyperref}
\usepackage{graphicx}
\usepackage{amsmath,amssymb,mathtools}
\usepackage{paralist}
\usepackage{bm}
\usepackage{xspace}
\usepackage{url}
\usepackage{boxedminipage}
\usepackage{wrapfig}
\usepackage{ifthen}
\usepackage{color}
\usepackage{xcolor}
\usepackage{framed}
\usepackage{fullpage}
\usepackage{subcaption}
\usepackage{thmtools}
\usepackage{thm-restate}

\newcommand{\cN}{\mathcal{N}}

\newcommand{\pV}[1]{\widetilde{V}^{(#1)}}

\newcommand{\eps}{\varepsilon}

\newcommand{\poly}{\mathrm{poly}}

\newcommand{\bu}{\boldsymbol{u}}
\newcommand{\bv}{\boldsymbol{v}}

\newcommand{\RR}{\mathbb{R}}

\newcommand{\argmin}{{\rm argmin}}

\newcommand{\kl}{KL}

\newcommand{\sm}{{\rm nsm}}

\renewcommand{\epsilon}{\varepsilon}
\newcommand{\Sec}[1]{\hyperref[sec:#1]{\S\ref{sec:#1}}} 
\newcommand{\Eqn}[1]{\hyperref[eq:#1]{(\ref{eq:#1})}} 
\newcommand{\Fig}[1]{\hyperref[fig:#1]{Fig.\,\ref{fig:#1}}} 
\newcommand{\Thm}[1]{\hyperref[thm:#1]{Theorem\,\ref{thm:#1}}} 
\newcommand{\Fact}[1]{\hyperref[fact:#1]{Fact\,\ref{fact:#1}}} 
\newcommand{\Lem}[1]{\hyperref[lem:#1]{Lemma\,\ref{lem:#1}}} 
\newcommand{\Prop}[1]{\hyperref[prop:#1]{Prop.~\ref{prop:#1}}} 
\newcommand{\Cor}[1]{\hyperref[cor:#1]{Corollary~\ref{cor:#1}}} 
\newcommand{\Conj}[1]{\hyperref[conj:#1]{Conjecture~\ref{conj:#1}}} 
\newcommand{\Def}[1]{\hyperref[def:#1]{Definition~\ref{def:#1}}} 
\newcommand{\Alg}[1]{\hyperref[alg:#1]{Alg.~\ref{alg:#1}}} 
\newcommand{\Ex}[1]{\hyperref[ex:#1]{Ex.~\ref{ex:#1}}} 
\newcommand{\Clm}[1]{\hyperref[clm:#1]{Claim~\ref{clm:#1}}} 

\usepackage{xcolor}

\newcommand{\ignore}[1]{}

\begin{document}

\newcommand\relatedversion{}
\newcommand{\lr}{\texttt{LR-{\allowbreak}Structural} }

\title{\Large Classic Graph Structural Features Outperform Factorization-Based Graph Embedding
  Methods on Community Labeling}
\author{Andrew Stolman\thanks{University of California, Santa Cruz. \tt{astolman@ucsc.edu}}
\and Caleb Levy\thanks{University of California, Santa Cruz. \tt{cclevy@ucsc.edu}}
\and C. Seshadhri\thanks{University of California, Santa Cruz. {\tt{sesh@ucsc.edu}}. Supported by NSF DMS-2023495, CCF-1740850, CCF-1813165, CCF-1839317, CCF-1908384, CCF-1909790, and ARO Award W911NF1910294.}
\and Aneesh Sharma\thanks{Google. \tt{aneesh@google.com}}}

\date{}

\maketitle

\begin{abstract} {\small\baselineskip=9pt 
  Graph representation learning (also called {\em graph embeddings}) is a popular technique for incorporating network structure into machine learning models. Unsupervised graph embedding methods aim to capture graph structure by learning a low-dimensional vector representation (the {\em embedding}) for each node. Despite the widespread use of these embeddings for a variety of downstream transductive machine learning tasks, there is little principled analysis of the effectiveness of this approach for common tasks. In this work, we provide an empirical and theoretical analysis for the performance of a class of embeddings on the common task of pairwise community labeling. This is a binary variant of the classic community detection problem, which seeks to build a classifier to determine whether a pair of vertices participate in a community. In line with our goal of foundational understanding, we focus on a popular class of unsupervised embedding techniques that learn low rank factorizations of a vertex proximity matrix (this class includes methods like GraRep, DeepWalk, node2vec, NetMF).  We perform detailed empirical analysis for community labeling over a variety of real and synthetic graphs with ground truth.  In all cases we studied, the models trained from embedding features perform poorly on community labeling.  In constrast, a simple logistic model with classic graph structural features handily outperforms the embedding models. For a more principled understanding, we provide a theoretical analysis for the (in)effectiveness of these embeddings in capturing the community structure. We formally prove that popular low-dimensional factorization methods either cannot produce community structure, or can only produce ``unstable" communities. These communities are inherently unstable under small perturbations. This theoretical result suggests that even though ``good" factorizations exist, they are unlikely to be found by computational methods.}
\end{abstract}

\section{Introduction}\label{sec:Introduction}

Graph structured data is ubiquitous. Capturing the graph structure is important for a wide variety of machine learning tasks, such as
ranking in social networks, content recommendations, and
clustering~\cite{easley2010networks}. A long-studied challenge for
building such machine learned models has been to capture the graph
structure for use in a variety of modeling tasks. \emph{Graph representation
  learning}, or {\em low-dimensional graph embeddings}, provide a
convenient solution to this problem.  Given a graph $G$ on $n$
vertices, these methods map each vertex to a vector in $\RR^d$, where
$d \ll n$, in an unsupervised or a self-supervised manner (it is also
sometimes referred to as a pre-training procedure). Typically, the
goal of the embedding is to represent graph proximity by (a function of) the dot product
of vectors, thereby implicitly giving a geometric
representation of the graph.\footnote{Since there is a wide range of
  methods for Graph representation learning, we refer the reader to
  the ``Shallow embeddings'' class in a recent
  survey~\cite{MLGSurvey20} for a more comprehensive overview.} The
dot product formulation provides a convenient form for building a
models (e.g. using deep learning). Moreover, the geometry of
the embedding allows efficient reverse-index lookups, using nearest neighbor search~\cite{covington2016deep,Twitter-embeddings}.

The
study of low-dimensional graph embeddings is an incredibly popular
research area, and has generated many exciting results over the past
few years (see surveys~\cite{HaYi18, MLGSurvey20} and a Chapter 23 in~\cite{pml1Book}). Nonetheless, there is 
limited 
principled understanding of the power of low-dimensional
embeddings (a few recent papers address this topic~\cite{SeSh20,CMST20,Loukas20,GJJ20}). 
Our work aims
to understand the effectiveness of a class of graph embeddings in
preserving graph structure as it manifests in performance on different downstream tasks.

Due to an explosion of interest in the area, there are by now a large
class of graph embedding methods~\cite{pml1Book}. For the sake of a principled study,
we focus on the important class of \emph{unsupervised low-rank factorization methods}.
While there do exist many methods outside this class, such factorization methods
cover a number of popular and influential
embeddings methods, including GraRep~\cite{OuCu16},
DeepWalk~\cite{PeAl14}, and Node2Vec~\cite{GrLe16}.  In fact, a recent result
shows that many existing embedding techniques can be recast as matrix
factorization methods~\cite{QiDo18}.  One begins with an $n \times n$
promixity matrix $M$, typically the adjacency matrix, the random walk
matrix, or some variant thereof (e.g. Node2Vec uses the matrix for a certain second order random walk). Using optimization techniques, the
matrix $M$ is approximated as a Gramian matrix $V^TV$, where
$V \in \RR^{d \times n}$. The column vectors of $V$ are the embeddings
of the vertices. In direct factorizations, one simply tries to
minimize $\|V^TV - M\|_2$. More sophisticated softmax factorizations
perform non-linear entry-wise transformations on $V^TV$ to approximate
$M$. This class of unsupervised embedding methods is among the most popular and prevalent low-dimensional graph
embeddings, and hence this is a particularly
useful class to quantify performance for.

Our aim is to study a natural question, albeit one that is somewhat challenging to pose
formally: \emph{to what extent do factorization-based embedding methods capture
graph structure relevant to downstream ML tasks?}

To this end, we fix the following well-defined \emph{pairwise
  community labeling} problem. Given two vertices $i$ and $j$, the
binary classification task is to determine whether they belong to the
same community. We note that this community labeling problem is an
instance of a broad range of community detection problems that have a
long history of study in the graph mining
literature~\cite{mmds_book}. We then attempt a rigorous theoretical
and empirical understanding of the performance of factorization-based
graph embeddings on community labeling.

We note that there are graph embedding methods that are not
factorization based (e.g. GraphSage~\cite{HaYi17}) as well as
factorization methods that do not use direct or softmax
factorizations~\cite{CMST20}, as well as GCNs and GNNs~\cite{KW17,
MLGSurvey20}. For the sake of a principled study, we chose a 
well-defined subclass of methods that covered many important
methods such as {\tt GraRep}~\cite{CaLu15}, {\tt DeepWalk}~\cite{PeAl14}, {\tt Node2Vec}~\cite{GrLe16}, and {\tt NetMF}~\cite{QiDo18}. Moreover, the recent
{\tt NetMF} algorithm shows how many past methods can be recast
as factorization methods.

One central promise of unsupervised graph embedding methods is to
preserve network structure in the geometry that can then be useful for
downstream tasks. In our work, we focus on the
task of community labeling both because it is directly connected to the core
question: how well do embeddings preserve the graph
structure? 

\subsection{Formal description of setting}\label{sec:setting}

We formally describe the graph embeddings techniques that are
studied in this work.  The learned factorization approach is to
approximate $M$ by the Gramian matrix $V^TV$ (the matrix of dot
products). Typically, this matrix $V$ is found by formulating a
machine learning problem, which has a loss function that minimizes a
distance/norm between $V^TV$ and $M$. Broadly speaking, we can
classify these methods into two categories:

\begin{asparaitem}
    \item {\bf Direct factorizations:} Here, we set $V$ as \\$\argmin_V \|V^TV - M\|_2$, where $M$ is typically
    (some power of) the graph adjacency matrix. Methods such as Graph Factorization, GraRep~\cite{CaLu15},
    and HOPE~\cite{OuCu16} would fall under this category.
    \item {\bf Softmax factorizations:} These methods factorize a stochastic matrix, such as (powers of)
    the random walk matrix. Since $V^TV$ is not necessarily stochastic, these methods apply the softmax
    to generate a stochastic matrix. Notable examples are such methods are DeepWalk~\cite{PeAl14} and Node2vec~\cite{GrLe16}.
    Formally, consider the normalized softmax matrix $\sm(V)$ given by 
\begin{equation}
\sm(V)_{ij} = \frac{\exp(\vec{v}_i \cdot \vec{v}_j)}{\sum_k \exp(\vec{v}_i \cdot \vec{v}_k)}
\end{equation}
Note that $\sm(V)$ is stochastic by construction. The objective of the
learning problem is to minimize 
$\kl(\sm(V), M)$, which is the sum of row-wise KL-divergence between the rows (this is equivalent to cross-entropy loss).
\end{asparaitem}

The recent {\tt NetMF}~\cite{QiDo18} method interpolates between these categories and shows
that a number of existing methods can be expressed as factorization methods, especially
of the above forms.
For this study, we only focus on the above two category of
unsupervised embedding methods. (We discuss this choice and other methods
in \Sec{just}.)

\begin{figure*}[h]
\begin{tabular}{c c c}
	\includegraphics[width=.33\linewidth]{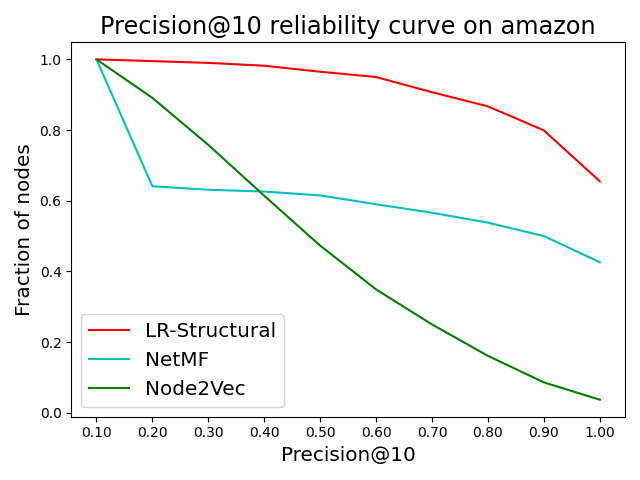} &
	\includegraphics[width=.33\linewidth]{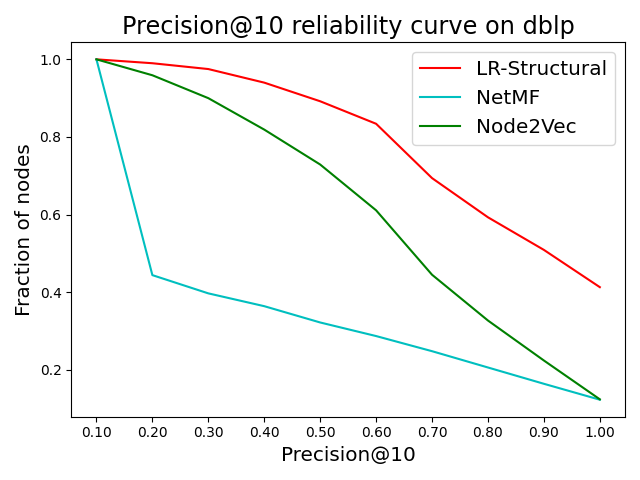}  &
	\includegraphics[width=.33\linewidth]{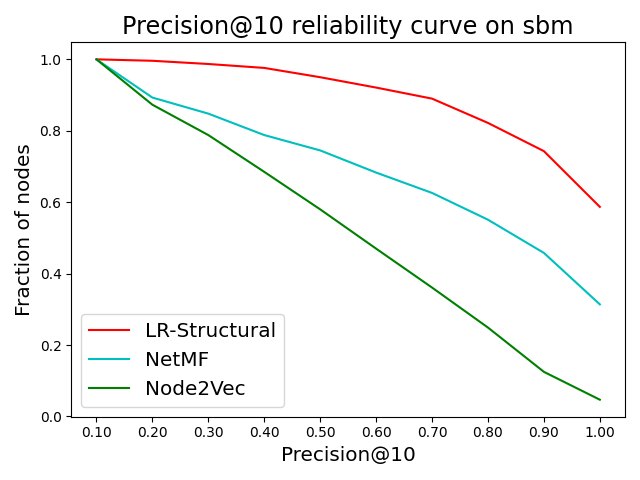}
\end{tabular}
\caption{Each point, $(x, y)$, on the curve represents the approximate fraction of vertices, $y$, for
	which the given method produces a precision@10 score of at least $x$. \lr is plotted against
the two best performing embedding methods. $1000$ vertices are sampled and for each vertex $v$ sampled,
the vertices of the graph
$u_1, \ldots, u_n$, are ordered by decreasing score assigned by the given classifier.
The precision@10 is the fraction of $u_1, \ldots, u_{10}$ which share a community with $v$.
}
\label{fig:p10-curves}
\end{figure*}

{\bf Empirical setup:} We empirically investigate performance of the
above methods on graphs where the ground truth communities correlate well with graph structure. 
In the case of real data, the ground truth is provided by the existing
community labels. With synthetic data, we explicitly construct
stochastic block models with well-defined communities. For every pair
of vertices $i, j$, the prediction problem is to determine whether
they belong to the same community (note that they may belong together
in multiple communities, but we do not require the community label
itself to be determined; just whether they belong in {\em any}
community together).

For both real and simulated data, we note that the ground truth is
sparse, i.e. the vast majority of node pairs do not belong to the same
community. Hence, it is appropriate to measure the prediction
performance using precision-recall curves for this highly imbalanced
label distribution~\cite{DG06}. Consistent with
most of the literature on graph embeddings and the applications that
often require nearest neighbor lookups, the main feature we use for
prediction is the value of the dot product~\cite{covington2016deep,MLGSurvey20}. 

{\bf Theoretical setup:} In order to analyze the performance, we
provide an abstraction of community structure from a matrix
standpoint: this can intuitively be thought of as having many dense
blocks in an overall sparse matrix. We then attempt to quantify "how
much" community structure can be present in a matrix $V^TV$ or
$\sm(V)$, for \emph{any} matrix $V \in \RR^{d\times n}$ (for
$d \ll n$). This formulation captures the fundamental notion of
a low-rank factorization, without referring to any specific method
to compute it. Our results hold for \emph{any} direct/softmax factorization method,
regardless of how the embedding is computed. Our formulation theoretically
investigates whether it is even possible to recreate community
structure using a low rank factorization of the form $V^TV$ or
$\sm(V)$.

\subsection{Main results}

All the graph embeddings methods we tested
(including {\tt GraRep}, {\tt DeepWalk}, {\tt Node2Vec}, and {\tt
  NetMF}) perform poorly on the community labeling task, and are
handily out-performed by a baseline logistic model \lr{} built using just
four classic graph structural features\footnote{The features for a
  node pair $(u,v)$ are: 1) Personalized-PageRank (PPR) score from $u$ to $v$, 2)
  PPR from $v$ to $u$, 3) cosine similarity among
  $N(u)$ and $N(v)$ (where $N(\cdot)$ denotes a node's neighborhood),
  and 4) the size of the cut separating $N(u)$ and $N(v)$}. We observe
the same outcome for a series of experiments on real data and
synthetic data.  Motivated by this empirical finding, we provide a
mathematical explanation for this result by providing a theorem which
shows that the community structure exhibited by softmax factorizations
is unstable under small perturbations of the embedding vectors.

{\bf Evaluations on real data:} In our experiments,
not only do we see poor {\em absolute} performance on the community
labeling task, but a baseline \lr{} based on ``classic'' graph features
handily outperforms the embeddings. We see this difference not only in
an overall manner, but across individual nodes in the graph. In
particular, Figure~\ref{fig:p10-curves} shows a ``reliability plot''
for precision@10 for a set of 1000 nodes sampled randomly from each of
the graphs. To produce this chart, we first randomly sampled 1000 nodes,
each of which has at least 10 neighbors in the same community. Then,
each of these vertices selects their top-10 predictions for a
same-community neighbor using a model, producing a distribution over
precision@10. Then for each model, one can produce a reliability plot
(see Figure~\ref{fig:p10-curves}) that produces points $(x,y)$: for any given value $x$ for
precision@10, define $y$ as the fraction of nodes that have
precision@10 of at least $x$. Thus,
we'd expect that a highly accurate model would have an almost flat curve.
The results in
Figure~\ref{fig:p10-curves} show that generally \lr{} can produce high
accurate ($\textrm{precision@k}\ge 0.7$) community labels for 
$20$--$40\%$ more nodes than the embedding-based models.
This suggests that the embeddings
can retrieve very few of the community neighbors, while classic graph
features used in \lr{} (PPR scores, neighborhood similarity) recover much of the community structure.

We emphasize that our empirical analysis is more nuanced than the more
common aggregated measurement (such as via the
AUC-ROC~\cite{MLGSurvey20}) as it measures the performance across
individual nodes in the graph. We believe this individual measurement
is more reflective of our goal, and also, as the label distribution
(which pairs co-occur in communities) is highly imbalanced, a P/R-like
curve provides a more useful measure than the ROC curves. The latter
can be misleading as an ROC curve can still look quite good while
misclassifying much of the minority class.~\cite{DG06}.

For completeness, we also perform experiments with regression models using the Hadamard product of the
two vectors. Again, the embeddings gives results of similar quality,
showing that linear functions (of the embedding vectors) fail to predict
community structure. Refer to Section~\ref{sec:data} for more details.

{\bf Evaluations on SBMs:} To further investigate this phenomenon,
we also generated synthetic graph instances
using Stochastic Block Models (SBMs) that have a very simple planted community structure. 
For example, we create a graph with $n = 10^5$ vertices and blocks of size $20$. The edge
density inside a block is $0.3$ and we connect the blocks by a sparse Erd\H{o}s-R\'{e}nyi graph
such that the average number of links within a block is equal to those that go between blocks.
We vary the overall edge density (while keeping the ratio between inter-block and intra-block edges constant)
and study the precision@10 scores. We observe that while the performance across methods tends to increase
with density, \lr{} still outperforms them.

{\bf Theoretical explanation:} We provide a formulation for what it means for a matrix
to exhibit community structure. We emphasize that this is not meant to
completely capture the challenging notion of communities (which has a
deep and rich history), but rather to give us some formal framework to state our impossibility 
results. Intuitively, an overall sparse matrix/graph has community structure if a non-trivial
fraction of the rows/vertices participate in small dense blocks of
entries. We then investigate
when $V^TV$ and $\sm(V)$ exhibit community structure. First, we prove that for $d \ll n$,
$V^TV$ cannot exhibit community structure, which provides a principled
justification for the empirical observation that direct factorization methods perform
poorly across all real and synthetic instances.

Interestingly, we also show that despite the negative result for
truncated-SVD in~\cite{SeSh20}, softmax factorizations \emph{can}
exhibit community structure similar to those given by the SBMs
discussed earlier. Recent work shows that threshold based sign factorizations can also embed 
such structure~\cite{CMST20}.
But we prove
that this community structure is fundamentally unstable under small
perturbations of the vectors obtained from softmax
factorizations. Meaning, if we take any matrix $V$ such that $\sm(V)$
has community structure, then with high probability,
$\sm(\widetilde{V})$ does not have such a structure (where
$\widetilde{V}$ is a slight random perturbation of $\sm(V)$).

This strongly suggests that optimization methods cannot produce low dimensional matrices $V$
where $\sm(V)$ has community structure. This theorem provides a mathematical understanding
of the limitations of softmax factorizations.
We do note that since softmax factorization are superior to direct factorizations,
since they avoid the direct impossibility for the latter.

\subsection{Related Work}\label{sec:just}

We briefly note a few important representation learning techniques
that are beyond the scope of our work. The most prominent among these
is the Graph Neural Networks such
as~\cite{HaYi17,KW17,VGCRLB18,HLGZLPL20}, which can be thought of as a
class of learned message passing methods. We refer the reader
to a nicely interpretable classification of these methods in a recent
survey~\cite{MLGSurvey20}. The factorization methods we study fall
under the "Shallow embeddings" classification there. 
There are several recent
works~\cite{GJJ20,Loukas20,XHLJ19} that theoretically study the power
of GNNs, but this is complementary to our work since we study a
different class of methods.

A recent result shows the inability of low-dimensional SVD based embeddings
of preserving the triangle structure of real-world networks~\cite{SeSh20}.
A followup showed that these impossibility results can be circumvented
by alternate embedding methods~\cite{CMST20}.
Our result can be thought of as a deeper investigation into this issue.
First, we look at a class of factorization methods subsuming those used in practice.
Secondly, we also focus on a specific downstream ML task, unless previous results
that focus solely on the graph structure.

\section{Mathematical results and interpretation} \label{sec:math}

We define a simple abstraction of community structure in a matrix $M$.
Then, we try to quantify how much community structure Gram matrices and softmax factorizations
can possess. We will state these as formal theorems, which are our main mathematical
result. The full proofs are given in Section~\ref{sec:proof}.

Let us start with an $n \times n$ matrix $M$ that represents the ``similarity" or likelihood of connection
between vertices. For convenience, let us normalize so that the $\forall i \in [n], \sum_{j \leq n} M_{i,j} \leq 1$.
(So the sum of similarities of a vertex is at most $1$.) A communities is essentially
a dense block of entries, which motivates the following definition. 
We use $\eps$ to denote a parameter for the threshold of community strength. One should think
of $\eps$ as a small constant, or something slowly decreasing in $n$ (like $1/\poly(\log n)$).

\begin{Definition} \label{def:comm} A pair of vertices $(i,j)$ is a 
\emph{potential community pair}
if both $M_{ij}$ and $M_{ji}$ are at least $\eps$.
\end{Definition}

Note that we do not expect all such pairs $(i,j)$ to truly be together in a community.
Hence, we only consider such a pair a potential candidate.
We expect community relationships to be mutual, even if the matrix $M$ is not. A community
can be thought of as a submatrix where at least a constant fraction of pairs are potential community
pairs. For our purposes, we do not need to further formalize. It is natural to expect
that $\Theta(n)$ pairs are community pairs; indeed, most vertices should participate
in communities, and will have at least a constant number of community neighbors. 
Our mathematical analyses shows that direct and softmax factorizations cannot produce
these many potential community pairs. 

{\bf Lower bound for direct factorizations:} We first show a strong lower bound for direct factorizations. We prove that the number
of potential community pairs in $V^TV$ is linear in the rank, and thus, a low-dimensional
factorization cannot capture community structure. The key insight is to use
the rotational invariance of Frobenius norms.

\begin{theorem} \label{thm:direct} Consider any matrix $V \in \RR^{d \times n}$
such that row sums in $V^TV$ have absolute value at most $1$. Then $V$ has at most
$d/2\eps^2$ potential community pairs.
\end{theorem}

\begin{proof} Since $V^TV$ has row sums of absolute value at most $1$, the
spectral radius (largest absolute value of eigenvalue) is also at most $1$. 
(This can be proven directly, but it also a consequence of the Gershgorin circle theorem~\cite{Ger}.)
Since the rank of $V^TV$ is at most $d$, $V^TV$ has at most $d$ non-zero eigenvalues.
We can express the Frobenius norm squared, $\|V^TV\|^2_2$, by the sums of squares
of eigenvalues. By the arguments above, $\|V^TV\|^2_2 \leq d$. 

Note that $\|V^TV\|^2_2$ is also the sums of squares of entries. Each potential community pair
contributes at least $2\eps^2$ to this sum. Hence, there can be at most $d/2\eps^2$
potential community pairs.
\end{proof}

{\bf The instability of softmax factorizations:} 
The properties of softmax factorizations are more nuanced. Firstly, we can prove
that softmax factorizations \emph{can} represent community structure quite effectively.

\begin{restatable}{theorem}{softmaxpos}\label{thm:softmaxpos} For $d = O(\log n)$, there exists $V \in \mathbb{R}^{d \times n}$ such
that $\sm(V)_{ij}$ exhibits community structure. Specifically, for any natural number $b \leq n$,
there exists $V \in \mathbb{R}^{d \times n}$  such that $\sm(V)$ has $n/b$ blocks
of size $b$, such that all entries within blocks are at least $1/2b$.
\end{restatable}

Indeed, this covers the various SBM settings we study, and demonstrates the superiority of
softmax factorizations for modeling community structure. We note that a similar theorem
(for a different type of factorization) was proved in~\cite{CMST20}.

On the other hand, we prove that these factorizations are highly \emph{unstable} to small
perturbations. Indeed, with a tiny amount of noise, any community pair can be destroyed with
high probability.

Formally, our noise model is as follows. Let $\delta > 0$ be a noise parameter.
Think of the $i$th column of $V$ as the $d$-dimensional vector $\vec{v}_i$, which
is the embedding of vertex $i$.
For every vector $\vec{v}_i$, we generate an independent random Gaussian $X_i \sim \cN(0,\delta^2)$
and rescale $\vec{v}_i$ as $(1+X_i)\vec{v}_i$ (formally, we rescale to $e^{X_i}\vec{v}_i$,
to ensure that the scaling is positive).
We denote this perturbed matrix as $\pV{\delta}$. We think of $\delta$ as a quantity going to zero, as $n$ becomes large.
(Or, one can consider $\delta$ as a tiny constant.) 

\begin{restatable}{theorem}{perturb} \label{thm:perturb} Let $c$ denote some absolute positive constant.
Consider any $V \in \mathbb{R}^{d \times n}$. 
For any $\delta > c\ln(1/\eps)/\ln n$, the following holds in $\sm(\pV{\delta})$. 
For at least $0.98n$ vertices $i$,
for any pair $(i,j)$, the pair is \emph{not} a potential community pair
with probability at least $0.99$.
\end{restatable}

Thus, with overwhelming probability, any community structure in $\sm(V)$ is destroyed by adding
$o(1)$ (asymptotic) noise. This is strong evidence that either noise in the input or numerical
precision in the final optimization could lead to destruction of community structure.
These theorems give an explanation of the poor performance of the embeddings methods studied.

\subsection{Proof ideas} \label{sec:proofidea}

In this section, we lay out the main
ideas in proving \Thm{perturb}. It helps to begin with the upper bound construction
of \Thm{softmaxpos}. Quite simply, we take $n/b$ random Gaussian vectors, and
map vertices in a community/block to the same vector. After doing some calculations
of random dot products, one can deduce that the vectors need to have length $\Omega(\sqrt{\ln n})$
for the construction to go through. By carefully look at the math, one also finds
that the construction is ``unstable". Even perturbing the vectors within a block slightly
(so that they are no longer the same vector) affects the block structure. We essentially
prove that these properties hold for \emph{any} set of vectors.

We outline the proof of \Thm{perturb}.
Suppose $\sm(V)_{ij} > \eps$. Note that $\sum_{k \in [n]} \sm(V)_{ik} = 1$, since
$\sm(V)$ is normalized by construction. Thus, there exists some $k$ such
that $\sm(V)_{ik} \leq 1/n$. We deduce that $\sm(V)_{ij}/\sm(V)_{ik} > \eps n$.
Writing out the entries and taking logs, this implies $\vec{v}_i \cdot \vec{v}_j - \vec{v}_i \cdot \vec{v}_k > \ln(\eps n)$. Therein lies the power (and eventual instability) of softmax 
factorizations: ratios of entries are transformed into differences of dot products. By
Cauchy-Schwartz,
one of $\vec{v}_i, \vec{v}_j, \vec{v}_k$ must have length $\Omega(\sqrt{\ln n})$ (ignoring
$\eps$ dependencies). By an averaging argument, we can conclude that for a vast majority
of community pairs $(i,j)$, $\|\vec{v}_i\|_2 = \Omega(\sqrt{\ln n})$. Note that
both $\sm(V)_{ij}$ and $\sm(V)_{ji}$ must be at least $\eps$; these quantities
have the same numerator, but different denominators. By analyzing these expressions
and some algebra, we can prove that $\Big|\|\vec{v}_i\|_2 - \|\vec{v}_j\|_2 \Big| = o(1/\sqrt{\ln n})$.

Thus, we discover the key property of communities expressed by softmax factorizations.
Vectors with a community have length $\Omega(\sqrt{\ln n})$, but the differences
in lengths must be $O(1/\sqrt{\ln n})$. Asymptotically, this is unstable to perturbations
in the length. A vanishingly small change in the vector lengths can destroy the community.

\section{Empirical verification}\label{sec:data}

\subsection{Community pair prediction}

We study the performance of matrix factorization embeddings at identifying community structure in a graph with many
possibly overlapping communities. Formally, we state this as a binary classification task over pairs of
vertices.  Every dataset consists of a graph, $G$, and a set of (possibly overlapping) communities,
$C_1, C_2, \ldots$.  This gives us a ground truth labeling over the pairs from $V \times V$ where positive
instances are those $(u, v)$ such that $u, v \in C_k$, for some
community $C_k$. We evaluate the performance of 
embedding techniques on this binary classification task over $V \times
V$ as follows:
\begin{asparaenum}
	\item An embedding method is applied to $G$ to obtain an embedding $\bv_1, \ldots, \bv_n$ of the nodes in
		$V$.
	\item A \emph{pair scoring function}, $f : V \times V
          \rightarrow \RR$, is constructed from the embedding of node pairs.
\end{asparaenum}
All of the evaluation is done with respect to this pair scoring function. We use this abstraction to encapsulate
all of the tasks downstream of the embedding generation itself. It consists of mapping the embedding vectors
to predictions in $\RR$. Ideally, $f(u, v) > f(x, y)$ whenever $u$ and $v$ share a community and $x$ and $y$ do not.
The construction of $f$ is based on the dot product of the embedding vectors of $u$ and $v$; details
are given in \Sec{features}. We stress that the vectors $\bv_1, \ldots, \bv_n$ incorporate information
about the neighborhoods. The underlying optimization of graph embeddings tries to produce close vectors 
for vertices in dense regions.

\subsection{Experimental results}

We observe that no choice of embedding method were competitive with
the baseline \lr method
on the task of community pair prediction. Across three datasets, \lr was consistently able
to identify community pairs while the embedding-based methods were not. For each method being compared, one thousand vertices
were selected at random from the graph, and we examine the precision of the classifier on the neighborhoods of each selected vertex. 
\Fig{p10-curves} shows that \lr makes more precise predictions on average than the best performing datasets, while \Fig{precision-table}
contains the mean precision for each classifier on each dataset.

The methods are evaluated by comparing the distribution of a precision@10 for each method on each dataset. 
For each of 1000 vertices sampled, $v$, 
we order the other vertices of the graph, $u_1, \ldots, u_n$ such that 
$f(\bv, \bu_i) \geq f(\bv, \bu_{i + 1})$ for all $i$. 
We compute the precision of the classifier on $(v, u_1), \ldots, (v, u_n)$ when it predicts positive labels
only for those $(v, u_i)$ such that $i \leq 10$. In other words, we sample a vertex at random and report
the fraction of its ten nearest neighbors in the embedding space with which it shares a community.

We represent the distribution of values of precision@10 scores as a reliability curve. This is the curve
$(x, y)$ such that at least a $y$ fraction of vertices sampled had a precision@10 score score of at least $x$.
Higher $y$ values for a given $x$ indicate better performance.
\Fig{p10-curves} contains the curves for the best performing dot product methods against the baseline,
while \Fig{precision-table} contains the mean across samples.

Given \Thm{perturb}, we expect that if $(u, v)$ is a community pair in some
ground truth matrix, $M$, then it is unlikely that $(u, v)$ is a community pair in any noisy approximation of $M$.
The results in this section bear out this conclusion. The community structure
of the original graphs are not preserved by the embeddings.

\begin{figure}[t]
	\begin{tabular}{| c | c | c | c |}
         	\hline
		method/dataset & amazon & dblp & sbm \\ \hline
		\lr & .92 & .80 & .88 \\ \hline
		DeepWalk & .44 & .49 & .35 \\ \hline
		NetMF & .49 & .28 & .66 \\ \hline
		Node2Vec & .49 & .61 & .55 \\ \hline
		DeepWalk-hp & .43 & .44 & .33 \\ \hline
		NetMF-hp & .41 & .37 & .76 \\ \hline
		Node2Vec-hp & .37 & .55 & .50 \\ \hline
	\end{tabular}
	\caption{Average precision@10 across 100 samples for all methods across the three datasets.
	In all cases, \lr is the best performing.}%
	\label{fig:precision-table}
\end{figure}
\subsection{Experimental setup}

\subsubsection{Community pair prediction methods}
We compare the performance of four different embedding methods (GraRep, DeepWalk, Node2Vec, NetMF) to a baseline to a simple supervised
logistic regression model using common structural graph features. 
The implementation of the embedding methods is based on
\cite{karateclub}. Except for setting the dimension to 128 for all
embedding methods, the hyperparameters are taken from
\cite{karateclub}. Our code is available at \url{https://github.com/astolman/snlpy}.
All our experiments were run on AWS using machines with up to 96 cores and 1 TB memory. Any method which did not complete in 24 hours
or which required more memory was considered not complete.

\begin{paragraph}{GraRep} This is a direct factorization method. The embedding it
	returns is the concatenation of an embedding fit to a modified version of the $k$-step random walk matrix
	for each $k$ up to a hyperparameter $K$. Consistent with \cite{karateclub}, we set $K = 5$. With this
	parameter setting, GraRep was not able to complete on the test machine since it requires explicitly computing
	powers of the adjacency matrix which is not feasible with any memory restraints.
\end{paragraph}

\begin{paragraph}{DeepWalk}
	This is a softmax factorization technique. It performs random walks to approximate the mean of the first $K$
	step random walk matrices (this is the window size parameter). In our experiments we use $K = 5$.
\end{paragraph}

\begin{paragraph}{Node2Vec}
	Like DeepWalk, this is a softmax factorization technique which generates an approximation to an average over
	the first $K$ random walk matrices by performing random walks.
	However, the random walks performed are a special kind which are controlled by parameters $p$ and $q$. When
	$p = q = 1$, the walks are identical to ordinary random walks. In our experiments, we use $p=q=0.5$ and $K = 5$.
\end{paragraph}

\begin{paragraph}{NetMF}
	This is a direct factorization technique which performs SVD on a modified version of the sum of the first $K$
	random walk matrices. For our experiments, and consistent with \cite{karateclub}, we use $K = 2$.
\end{paragraph}

\begin{paragraph}{\lr} For the baseline method, we use a logistic regression model trained to predict
	community pairs based on four tractable graph features found to be useful in the literature:
	\begin{asparaitem}
	\item Cosine similarity between $u$'s and $v$'s adjacency vectors \cite{SSG17},
	\item Size of the cut between $u$'s neighborhood and $v$'s neighborhood, and
	\item Personalized PageRank (PPR) score from $u$ to $v$, as well as that from $v$ to $u$ \cite{ACL07}.
	\end{asparaitem}
	We use the approximation
	algorithm from \cite{ACL07} to compute sparse approximate PPR vectors. Note that this allows for all features
	to be computed locally, i.e. the space/time complexity for computing the features for one vertex pair are
	independent of graph size. The overall space/time costs in practice are on par with the embedding methods.
\end{paragraph}

\subsubsection{Pair features}\label{sec:features}

The embedding methods produce a feature vector for each vertex. Since the community pair prediction task requires a
set of feature vectors over $V \times V$ and a method to produce scores, we need to
turn node features into pair features.  Building on strategies
proposed in prior work~\cite{GrLe16}, for the embeddings we choose
two methods from the literature which are computationally efficient, and relatively accurate. Given a $d$-dimensional
embedding, $\bv_1, \ldots, \bv_n$ we compute the feature vector for $(u, v) \in V \times V$ with the \emph{dot product}
and \emph{Hadamard product} denoted $\bu \cdot \bv$ and $\bu \circ \bv$ respectively.

The dot product is the usual inner product over $\RR^d$, i.e. $\bu \cdot \bv = \sum_{i \in [d]}\bu(i)\bv(i)$. It
takes the two $d$-dimensional vector embeddings and maps them to a $1$-dimensional feature vector. On the other
hand, the Hadamard product of $\bu$ and $\bv$, $\bu \circ \bv$, is a $d$-dimensional vector whose $i$th coordinate
is the product of the $i$th coordinates of $\bu$ and $\bv$.

We will indicate when the pair features used are the dot product or Hadamard product of the endpoints. The actual
scoring function, $f: V \times V \rightarrow \RR$, we use is implicit from the pair features. Whenever the pair
features are dot products, $f(u, v) = \bu \cdot \bv$. Whenever the pair features are the Hadamard product of
$\bu$ and $\bv$, we fit a logistic regression model to the features $\bu \circ \bv$ to produce a weight vector
$\beta$, and so the pair scoring function becomes $f(u, v) = \sigma(\beta \cdot (\bu \circ \bv))$ where
$\sigma = \exp(x)/(1 + \exp(x))$ is the sigmoid function.

In order to compute the weight vector, $\beta$, in the Hadamard product case, we select a training set of size
$50n$ by choosing 50 vertices which
participate in a community with at least twenty members, $v_1, \ldots, v_{50}$, and compute $\bu \circ \bv_i$
for each $u \in V$. This collection of $50n$ feature vectors, along with ground truth labeling of community pairs,
are given to an sklearn LogisticRegression object which produces $\beta$.

\begin{figure*}[h]
	\begin{tabular}{c c c}
	\includegraphics[width=.33\linewidth]{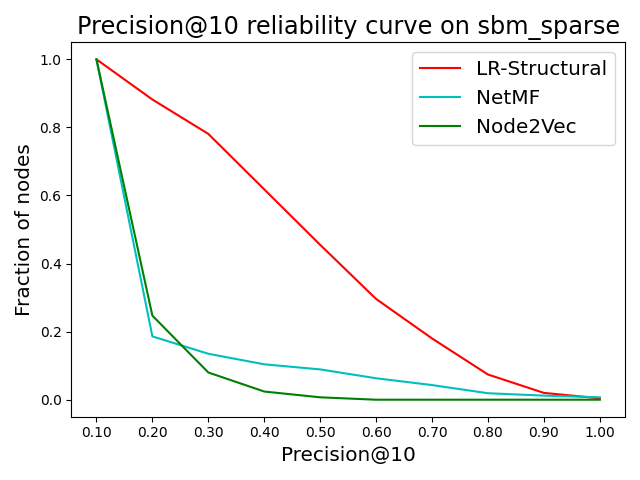} &
	\includegraphics[width=.33\linewidth]{figures/sbm_curves.png}
	\includegraphics[width=.33\linewidth]{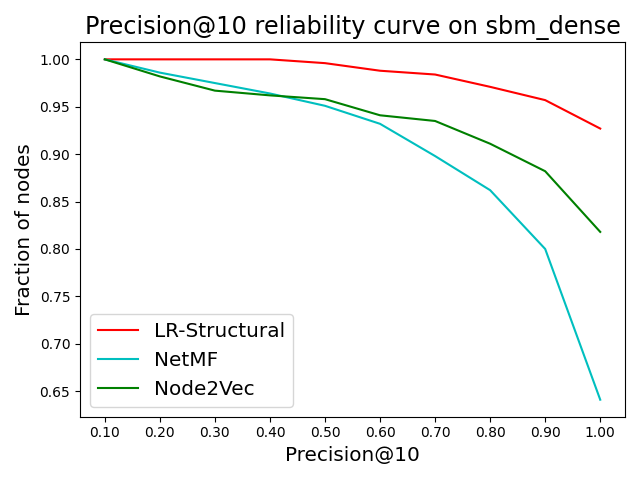}  &
\end{tabular}
	\caption{Precision@10 reliability curves for sbm datasets. All datasets are stochastically generated datasets with size 10,000
	and disjoint communities of size 20. The average degrees are 4, 12 and 20 from left to right.
	Curves are generated from 100 samples.}%
	\label{fig:sbms}
\end{figure*}\subsubsection{Datasets}

We show the performance of the various embedding methods contrasted with \lr on three datasets: two publicly available
real world datasets with ground truth community labels, and synthetic stochastic block models (SBM).
\begin{asparaitem}
	\item dblp: a co-authorship network of 317K computer science authors with
		communities defined as venues
	\item amazon: a network of 335K products on amazon with a link representing
		frequent co-purchasing. Communities are product categories.
	\item sbm: synthetic dataset with 100K vertices and communities of size 20. Edges are randomly generated
		with an inter-community edge probability of 0.3 and intra-community edge probability of $0.3/n$.
\end{asparaitem}

\begin{table}[h]
	\centering
\begin{tabular} {| c | c | c | c |}
	\hline
	\bf Statistic	& \bf dblp	& \bf amazon	& \bf sbm \\
	\hline
	\makecell{number \\of nodes} & 317K	& 334K		& 100K \\
	\hline
	\makecell{number \\of edges}	& 1M	& 1M		& 585K \\
	\hline
	\makecell{num \\communites}	& 13K	& 75K		& 5K \\
	\hline
	max size	& 75K	& 53K		& 20 \\
	\hline
	median size	& 8	& 5		& 20 \\
	\hline
	\makecell{mean comm \\density} & 0.09 & 0.10 & 0.30 \\
	\hline
\end{tabular}
\caption{Dataset summary}
\medskip
\footnotesize
{\bf
``Mean comm density'' is an average weighted by community size of the edge density of each community. 
 \label{tab:datasets}
}
\end{table}

\section{Stochastic block models}

We complement our experimental results on real datasets with measuring
performance of graph embedding methods on synthetic datasets generated according to the popular Stochastic Block
Model (SBM). 

An $n$-vertex graph, $G$, is generated according to an SBM by
partitioning the vertex set, $[n]$ into $k$ equal sized communities, $C_1, C_2, \ldots,
C_k$. The distribution of edges is controlled by two parameters.
If two vertices are in the same community, an edge is added with probability $p$.
If not, then an edge is added with probability $q$.

We chose the SBM parameters in order
to approximate some of the features we observed in the empirical datasets.
All SBMs have one hundred
thousand vertices and communities (or ``blocks'') of size 20. The parameter $q$ is set so
that the average number of neighbors a vertex has inside its community
is equal to that outside of its community. In real data, there are often 3 or 4 times as many
inter-community neighbors as intra-community neighbors.  
The three SBM graphs generated, sbm\_sparse, sbm, sbm\_dense, have average degrees
4, 12 and 20 respectively. 
 
Motivated by \Thm{softmaxpos}, we can study the embeddings' resilience
to noise in the simulated setting.  \Fig{sbms} shows the results from this experiment. Note that the
accuracy of the embeddings at pairwise community labeling decreases as
the internal density of the communities decreases.
In all cases, \lr{} is able to outperform the factorization-
based embedding methods across a regime of parameters commonly encountered in sparse datasets. 

\section{Proof of instability in softmax factorizations} \label{sec:proof}

We restate and prove the upper bound theorem, showing that softmax factorizations
can recreate community structure.

\softmaxpos

%

\begin{proof} We apply the probabilistic method. We select $V$ from a random distribution, and 
prove that the desired community structure is exhibited with high probability.

Let $d = c\ln n$, for some sufficiently large constant $c$. We will construct $n/b$ blocks of vertices,
each with $b$ vertices. All the vertices in the $b$th block will be represented by the same vector $\vec{v}_b$.
We will set each $\vec{v}_b$ to be a uniform random Gaussian vector of length $2\sqrt{\ln n}$. Thus, for any two vertices
$i,j$ within a block, $\vec{v}_i \cdot \vec{v}_j = 4\ln n$. For two vertices $i,j$ that occur in different blocks,
the dot product $\vec{v}_i \cdot \vec{v}_j$ is normally distributed with mean zero
and variance $(4\ln n)/d = 4/c$.
By tail bounds for the Gaussian, $\Pr[\vec{v}_i \cdot \vec{v}_j \geq (16\ln n)/c] \leq 1/n^4$,
where the probability is over the choice of the vectors. By a union bound over all (at most) $n^2$ pairs,
with probability at least $1-1/n^2$, for every pair $i, j$ in distinct blocks,
$\vec{v}_i \cdot \vec{v}_j < (16\ln n)/c$, which is at most $\ln n$ (for sufficiently large $c$).

Note that, for all vertices $i$, we can split the sum $\sum_k \exp(\vec{v}_i \cdot \vec{v}_k)$
into vertices within $i$'s block and outside $i$'s block. The total sum outside the block
is at most $n \times \exp(\ln n) = n^2$. The sum within the block is $b \exp(4\ln n) = bn^4$.
For $i,j$ within a block, $\sm(V)_{ij} \geq \exp(4\ln n)/(bn^4 + n^2) \geq 1/2b$. Thus, $\sm(V)$
exhibits the desired community structure.
\end{proof}

We note some features of this construction. Firstly, the vectors are of non-constant length.
Moreover, the vectors within a community are extremely close to each other compared to their length; 
in the construction, they are actually identical.

We essentially prove that both these properties are necessary for \emph{any} $V$
where $\sm(V)$ exhibits community structure. The second property of closeness is extremely sensitive
to noise, and we prove that even small amounts of noise can destroy the community structure.
In practice, we believe that such unstable solutions are hard to find, and the ambient
noise in data forces solutions that are stable to perturbations. This point is validated
in our experiments, where the solutions do not exhibit community structure.

\begin{lemma} \label{lem:length} For at least $0.99n$ vertices $i$, the following
property holds. If $(i,j)$ is a potential community pair in $\sm(V)$,
then $\|\vec{v}_i\|_2 \geq \sqrt{\ln (\eps n/100)}/2$ and $\Big |\|\vec{v}_i\|_2 - \|\vec{v}_j\|_2\Big| \leq 2\ln(1/\eps)/\sqrt{\ln (\eps n/100)}$.
\end{lemma}

\begin{proof} 
Recall that $\sm(V)_{i,j} = \exp(\vec{v}_i \cdot \vec{v}_j)/\sum_k \exp(\vec{v}_i \cdot \vec{v}_k)$.
There are $n$ vertices, so by Markov's inequality, for at least $0.99n$ vertices $k$, $\sm(V)_{i,k} \leq 100/n$.
Pick any such $k$, and note that $\sm(V)_{i,j}/\sm(V)_{i,k} = \exp(\vec{v}_i \cdot \vec{v}_j - \vec{v}_i \cdot \vec{v}_k)
\geq \eps n/100$. Taking logs, 
\begin{eqnarray*}
    \vec{v}_i \cdot (\vec{v}_j - \vec{v}_k) \geq \ln(\eps n/100) \\
    \Longrightarrow \|\vec{v}_i\|_2 \|\vec{v}_j - \vec{v}_k\|_2 \geq \ln(\eps n/100)  \\
    \Longrightarrow \|\vec{v}_i\|_2 (\|\vec{v}_j\|_2 + \|\vec{v}_k\|_2) \geq \ln(\eps n/100) 
\end{eqnarray*}
(The first implication comes from the Cauchy-Schwartz inequality, and the second comes
from the triangle inequality.)

For convenience, let $N := \eps n/100$.
Suppose that $\|\vec{v}_i\|_2 < \sqrt{\ln N}/2$ and $\|\vec{v}_j\|_2 < \sqrt{\ln N}/2$.
This implies that $\|\vec{v}_j\|_2 + \|\vec{v}_k\|_2 \geq \sqrt{\ln N}$, and thus,
$\|\vec{v}_k\|_2 \geq \sqrt{\ln N}/2$. 

Thus, one of the following holds. Either for every potential community pair $(i,j)$, 
$\max(|\vec{v}_i\|_2, \|\vec{v}_j\|_2) \geq \sqrt{\ln N}/2$, or
for at least $0.99n$ vertices $k$, $\|\vec{v}_k\|_2 \geq \sqrt{\ln N}/2$.
Regardless, for at least $0.99n$ vertices $i$, for every potential community pair $(i,j)$,
$\max(|\vec{v}_i\|_2, \|\vec{v}_j\|_2) \geq \sqrt{\ln N}/2$.

Consider a potential community pair $(i,j)$ where
$\max(|\vec{v}_i\|_2, \|\vec{v}_j\|_2) \geq \sqrt{\ln N}/2$. Wlog, let
$\|\vec{v}_i\|_2$ be larger. Recall that
$\sm(V)_{ij} = \exp(\vec{v}_i \cdot \vec{v}_j)/\sum_k \exp(\vec{v}_i \cdot \vec{v}_k) \geq \eps$.
\begin{eqnarray}
    & & \exp(\vec{v}_i \cdot \vec{v}_j) \geq \eps \sum_k \exp(\vec{v}_i \cdot \vec{v}_k) \nonumber \\
    & \Longrightarrow & \exp(\vec{v}_i \cdot \vec{v}_j) \geq \eps \exp(\vec{v}_i \cdot \vec{v}_i) \nonumber \\
    & \Longrightarrow & \vec{v}_i \cdot \vec{v}_j \geq \ln \eps + \|\vec{v}_i\|^2_2 \nonumber \\
    & \Longrightarrow & \|\vec{v}_i\|_2 \|\vec{v}_j\|_2 \geq \ln \eps + \|\vec{v}_i\|^2_2 \ \ \textrm{(Cauchy-Schwartz)} \nonumber \\
    & \Longrightarrow & \|\vec{v}_j\|_2 \geq \|\vec{v}_i\|_2 - (\ln 1/\eps)/\|\vec{v}_i\|_2 \nonumber
\end{eqnarray}
Since $\|\vec{v}_i\|_2 \geq \sqrt{\ln N}/2$ (and $\|\vec{v}_i\|_2 \geq \|\vec{v}_j\|_2$), 
$|\|\vec{v}_i\|_2 - \|\vec{v}_j\|_2| \leq 2\ln(1/\eps)/\sqrt{\ln N}$
\end{proof}

We now state the perturbation instability theorem, which should be intuitively
clear by Lemma~\ref{lem:length}. For (almost all) community pairs $(i,j)$,
both $\|\vec{v}_i\|_2$ and $\|\vec{v}_j\|_2$ are at least $\Omega(\sqrt{\ln n})$,
but the \emph{difference} between the lengths is $O(1/\sqrt{\ln n})$.
Infinitesimal perturbations will destroy such a property, and thus the community
pair will be lost.

\perturb

\begin{proof} Let us apply Lemma~\ref{lem:length} to both $V$ and $\pV{\delta}$. By the union
bound, there are at least $0.98n$ vertices $i$ that satisfy the property of
the lemma for both $V$ and $\sm(V)$. Consider any such $i$. Recall
that $\pV{\delta}$ is formed by perturbing each $\vec{v}_i$ by $e^{X_i} \vec{v}_i$,
where $X_i \sim \cN(0,\delta^2)$. For convenience, let $\vec{v'_i}$ denote
the perturbed vectors. By the properties of Lemma~\ref{lem:length},
for $(i,j)$ to be a community pair,

$$ \Big |\|\vec{v'_i}\|_2 - \|\vec{v'_j}\|_2\Big| \leq 2\ln(1/\eps)/\sqrt{\ln (\eps n/100)} $$

For convenience, set $\alpha := 2\ln(1/\eps)/\sqrt{\ln (\eps n/100)}$.
We conclude from above that
\begin{eqnarray*}
\|\vec{v'_j}\|_2 = e^{X_j} \|\vec{v}_j\|_2 \in [\|\vec{v'_i}\|_2 - \alpha, \|\vec{v'_i}\|_2 + \alpha] \\
\Longrightarrow X_j + \ln \|\vec{v}_j\|_2 \in [\ln(\|\vec{v'_i}\|_2 - \alpha), \ln(\|\vec{v'_i}\|_2 + \alpha)]
\end{eqnarray*}

By Lemma~\ref{lem:length} applied to $\pV{\delta}$, $\|\vec{v'_i}\|_2 \geq \sqrt{\ln(\eps n/100)}/2$. For convenience, set $\beta := \|\vec{v'_i}\|_2$.
Note that $\ln(\beta - \alpha) = \ln \beta + \ln(1-\alpha/\beta)$. Since $\alpha/\beta \leq 4\ln(1/\eps)/\ln(\eps n/100) < 1$,
we can lower bound $\ln(1-\alpha/\beta) \geq -2\alpha/\beta$. Similarly, we can upper bound $\ln(\beta+\alpha)$ by $\ln \beta + \ln(1+\alpha/\beta) \leq \ln \beta + 2\alpha/\beta$.

Hence, $X_j + \ln \|\vec{v}_j\|_2$ lies in an interval of size at most $4\alpha/\beta \leq 16\ln(1/\eps)/\ln(\eps n/100)$. 
Recall that $X_j$ is a Gaussian distributed as $\cN(0,\delta^2)$. For $\delta > c\ln(1/\eps)/\ln(\eps n/100)$,
we need $X_j$ to lie in an interval of size at most $16/c$ times the standard deviation ($\delta$). 
By the properties of the Gaussian, this probability is at most $0.01$, for some sufficiently large constant $c$.
Thus, the pair $(i,j)$ will be a community pair with probability at most $0.01$.
\end{proof}

\bibliographystyle{alpha}
\bibliography{embeddings}

\newcommand{\etalchar}[1]{$^{#1}$}
\begin{thebibliography}{CAEHP{\etalchar{+}}20}

\bibitem[CAEHP{\etalchar{+}}20]{MLGSurvey20}
Ines Chami, Sami Abu-El-Haija, Bryan Perozzi, Christopher R{\'e}, and Kevin
  Murphy.
\newblock Machine learning on graphs: A model and comprehensive taxonomy.
\newblock {\em arXiv:2005.03675}, 2020.

\bibitem[CAS16]{covington2016deep}
Paul Covington, Jay Adams, and Emre Sargin.
\newblock Deep neural networks for youtube recommendations.
\newblock In {\em Proceedings of the 10th ACM conference on recommender
  systems}, pages 191--198, 2016.

\bibitem[CLX15]{CaLu15}
Shaosheng Cao, Wei Lu, and Qiongkai Xu.
\newblock {GraRep}: Learning graph representations with global structural
  information.
\newblock In {\em Conference on Information and Knowledge Management (CIKM)},
  pages 891--900. {ACM} Press, 2015.

\bibitem[CMST20]{CMST20}
Sudhanshu Chanpuriya, Cameron Musco, Konstantinos Sotiropoulos, and
  Charalampos~E Tsourakakis.
\newblock Node embeddings and exact low-rank representations of complex
  networks.
\newblock In {\em NeurIPS}, 2020.

\bibitem[DG06]{DG06}
Jesse Davis and Mark Goadrich.
\newblock The relationship between precision-recall and roc curves.
\newblock In {\em International Conference on Machine Learning (ICML)}, pages
  233--240, 2006.

\bibitem[EK10]{easley2010networks}
David Easley and Jon Kleinberg.
\newblock {\em Networks, crowds, and markets}, volume~8.
\newblock Cambridge university press Cambridge, 2010.

\bibitem[Ger20]{Ger}
Gershgorin circle theorem.
\newblock \url{https://en.wikipedia.org/wiki/Gershgorin_circle_theorem}, 2020.

\bibitem[GJJ20]{GJJ20}
Vikas~K Garg, Stefanie Jegelka, and Tommi Jaakkola.
\newblock Generalization and representational limits of graph neural networks.
\newblock {\em arXiv:2002.06157}, 2020.

\bibitem[GL16]{GrLe16}
Aditya Grover and Jure Leskovec.
\newblock node2vec: Scalable feature learning for networks.
\newblock In {\em Conference on Knowledge Discovery and Data Mining (KDD)},
  pages 855--864. {ACM}, 2016.

\bibitem[HLG{\etalchar{+}}20]{HLGZLPL20}
Weihua Hu*, Bowen Liu*, Joseph Gomes, Marinka Zitnik, Percy Liang, Vijay Pande,
  and Jure Leskovec.
\newblock Strategies for pre-training graph neural networks.
\newblock In {\em International Conference on Learning Representations}, 2020.

\bibitem[HYL17]{HaYi17}
Will Hamilton, Zhitao Ying, and Jure Leskovec.
\newblock Inductive representation learning on large graphs.
\newblock In {\em Neural Information Processing Systems ({NeurIPS})}, page~11,
  2017.

\bibitem[HYL18]{HaYi18}
William~L Hamilton, Rex Ying, and Jure Leskovec.
\newblock Representation learning on graphs: Methods and applications.
\newblock page~24, 2018.

\bibitem[Lou20]{Loukas20}
Andreas Loukas.
\newblock What graph neural networks cannot learn: depth vs width.
\newblock In {\em International Conference on Learning Representations}, 2020.

\bibitem[LRU20]{mmds_book}
Jure Leskovec, Anand Rajaraman, and Jeffrey~David Ullman.
\newblock {\em Mining of massive data sets}.
\newblock Cambridge university press, 2020.

\bibitem[Mur21]{pml1Book}
Kevin~P. Murphy.
\newblock {\em Probabilistic Machine Learning: An introduction}.
\newblock MIT Press, 2021.

\bibitem[OCP{\etalchar{+}}16]{OuCu16}
Mingdong Ou, Peng Cui, Jian Pei, Ziwei Zhang, and Wenwu Zhu.
\newblock Asymmetric transitivity preserving graph embedding.
\newblock In {\em Conference on Knowledge Discovery and Data Mining (KDD)},
  pages 1105--1114. {ACM}, 2016.

\bibitem[PARS14]{PeAl14}
Bryan Perozzi, Rami Al-Rfou, and Steven Skiena.
\newblock {DeepWalk}: Online learning of social representations.
\newblock In {\em Conference on Knowledge Discovery and Data Mining (KDD)},
  pages 701--710. {ACM} Press, 2014.

\bibitem[QDM{\etalchar{+}}18]{QiDo18}
Jiezhong Qiu, Yuxiao Dong, Hao Ma, Jian Li, Kuansan Wang, and Jie Tang.
\newblock Network embedding as matrix factorization: Unifying {DeepWalk},
  {LINE}, {PTE}, and node2vec.
\newblock In {\em Conference on Web Science and Data Mining (WSDM)}. {ACM}
  Press, 2018.

\bibitem[RAL07]{ACL07}
Fan~Chung Reid~Andersen and Kevin Lang.
\newblock Using pagerank to locally partition a graph.
\newblock {\em Internet Mathematics}, 4:35--64, 2007.

\bibitem[RKS20]{karateclub}
Benedek Rozemberczki, Oliver Kiss, and Rik Sarkar.
\newblock {Karate Club: An API Oriented Open-source Python Framework for
  Unsupervised Learning on Graphs}.
\newblock In {\em Conference on Information and Knowledge Management (CIKM)}.
  ACM, 2020.

\bibitem[SSG17]{SSG17}
Aneesh Sharma, C.~Seshadhri, and Ashish Goel.
\newblock When hashes met wedges: A distributed algorithm for finding high
  similarity vectors.
\newblock In {\em Conference on the World Wide Web (WWW)}, page 431–440,
  2017.

\bibitem[SSSG20]{SeSh20}
C.~Seshadhri, Aneesh Sharma, Andrew Stolman, and Ashish Goel.
\newblock The impossibility of low-rank representations for triangle-rich
  complex networks.
\newblock {\em Proceedings of the National Academy of Sciences},
  117(11):5631--5637, 2020.

\bibitem[TW17]{KW17}
N~Kipf Thomas and Max Welling.
\newblock Semi-supervised classification with graph convolutional networks.
\newblock {\em International Conference on Learning Representations}, 2, 2017.

\bibitem[Twi18]{Twitter-embeddings}
Embeddings@twitter.
\newblock
  \url{https://blog.twitter.com/engineering/en_us/topics/insights/2018/embeddingsattwitter.html},
  2018.

\bibitem[VCC{\etalchar{+}}18]{VGCRLB18}
Petar Veličković, Guillem Cucurull, Arantxa Casanova, Adriana Romero, Pietro
  Liò, and Yoshua Bengio.
\newblock Graph attention networks.
\newblock In {\em International Conference on Learning Representations}, 2018.

\bibitem[XHLJ19]{XHLJ19}
Keyulu Xu, Weihua Hu, Jure Leskovec, and Stefanie Jegelka.
\newblock How powerful are graph neural networks?
\newblock In {\em International Conference on Learning Representations}, 2019.

\end{thebibliography}

\end{document}